\documentclass[12pt, amsmath, reqno]{amsart}
\usepackage{amsthm}
\usepackage{latexsym}
\usepackage{amsfonts}
\usepackage{amssymb}
\usepackage{color}
\usepackage{bbm,dsfont}
\usepackage{graphicx}
\usepackage{epstopdf}
\usepackage{hyperref}
\usepackage{enumerate}

\newtheorem{lemma}{Lemma}

\newtheorem{proposition}{Proposition}

\theoremstyle{definition}
\newtheorem{example}{Example}

\newcommand{\R}{\mathbb{R}}

\newcommand{\C}{\mathbb C} 
\newcommand{\integer}{\mathbb Z} 
\newcommand{\half}{\tfrac{1}{2}} 
\newcommand{\mo}[1]{\left| #1 \right|} 

\newcommand{\hi}{\mathcal{H}} 
\newcommand{\ki}{\mathcal{K}} 
\newcommand{\K}{\mathcal{K}} 
\newcommand{\lh}{\mathcal{L(H)}} 
\newcommand{\lk}{\mathcal{L(K)}} 
\newcommand{\trh}{\mathcal{T(H)}} 
\newcommand{\trhk}{\mathcal{T}(\mathcal{H}\otimes \mathcal{K})} 
\newcommand{\sh}{\mathcal{S(H)}} 
\newcommand{\ph}{\mathcal{P(H)}} 
\newcommand{\sk}{\mathcal{S(K)}} 
\newcommand{\uh}{\mathcal{U(H)}} 

\newcommand{\ip}[2]{\left\langle\,#1\,|\,#2\,\right\rangle} 
\newcommand{\<}{\langle} 
\renewcommand{\>}{\rangle} 
\newcommand{\kb}[2]{|#1\rangle\langle#2|} 
\newcommand{\tr}[1]{\textrm{tr}\left[#1\right]} 
\newcommand{\ptr}[1]{\textrm{tr}_{\mathcal{K}}[#1]} 
\newcommand{\id}{\mathbbm{1}} 

\newcommand{\supp}[1]{\textrm{supp}[#1]} 
\newcommand{\mc}[1]{\mathcal{#1}} 

\newcommand{\Ao}{\mathsf{A}}
\newcommand{\A}{\mathsf{A}}
\newcommand{\Bo}{\mathsf{B}}
\newcommand{\Eo}{\mathsf{E}}
\newcommand{\E}{\mathsf{E}}
\newcommand{\Fo}{\mathsf{F}}
\newcommand{\Zo}{\mathsf{Z}}
\newcommand{\Z}{\mathsf{Z}}

\begin{document}

\title[]{Limitations on post-processing assisted quantum programming}
\author[Heinosaari]{Teiko Heinosaari}
\author[Miyadera]{Takayuki Miyadera}
\author[Tukiainen]{Mikko Tukiainen}
\address{\textbf{Teiko Heinosaari}; Turku Centre for Quantum Physics, Department of Physics and Astronomy, University of Turku, Finland}
\email{teiko.heinosaari@utu.fi}
\address{\textbf{Takayuki Miyadera}; Department of Nuclear Engineering, Kyoto University - 6068501 Kyoto, Japan}
\email{miyadera@nucleng.kyoto-u.ac.jp}
\address{\textbf{Mikko Tukiainen}; Turku Centre for Quantum Physics, Department of Physics and Astronomy, University of Turku, Finland}
\email{mikko.tukiainen@utu.fi}

\maketitle

\begin{abstract}
A quantum multimeter is a programmable device that can implement measurements of different observables depending on the programming quantum state inserted into it. The advantage of this arrangement over a single purpose device is in its versatility: one can realize various measurements simply by changing the programming state. The classical manipulation of measurement output data is known as post-processing. In this work we study the post-processing assisted quantum programming, which is a protocol where quantum programming and classical post-processing are combined. We provide examples showing that these two processes combined can be more efficient than either of them used separately. Furthermore, we derive an inequality relating the programming resources to their corresponding programmed observables, thereby enabling us to study the limitations on post-processing assisted quantum programming.

\end{abstract}

\section{Introduction}\label{sec:intro}

A programmable device can operate in different ways depending on the instructions inserted into it.
A quantum device may take these instructions not only in the form of classical data, but also in the form of a quantum state \cite{NiCh97}. 
If distinct programs correspond to distinguishable quantum states, then the situation is no different from classical programming.
Yet the setting changes when the instructions correspond to non-orthogonal pure quantum states; this may allow one to do something that is not possible with classical programming.
Using quantum resources instead of their classical counterparts in programming a device may have advantages, as in the case of a semiquantum Bell scenario \cite{Buscemi12}.

Any quantum measurement can be viewed as a process where the system of interest is brought into contact with an initially uncorrelated apparatus and, after an interaction, the value of the measured quantity is read from the apparatus' pointer scale \cite{Ozawa84}. A quantum measurement set-up may be used as a programmable device by considering the state of the apparatus as a program \cite{ZiBu05}.
This kind of a quantum multimeter can implement measurements of different observables even if the other parts of the device, except the programming state, are kept fixed.
However, it is known that for two different sharp quantum observables the corresponding programming states must be orthogonal \cite{DaPe05,HeTu15} and hence, in this case, there is no quantum advantage over classical programming.

In another related scenario an experimenter does not change any part of a measurement device, but is allowed to perform classical operations to post-process the measurement data.
In this way one can implement sets of quantum observables that are jointly measurable \cite{AlCaHeTo09}.
The limitation in the case of sharp observables is that only commuting sharp observables are jointly measurable \cite{Lahti03}.

Interestingly, it was observed in \cite{ZiBu05} that by allowing both quantum programming and classical post-processing it is possible to overcome the limitations that each procedure possesses when performed separately.
Indeed, it was demonstrated that there exists a programmable device that implements, after post-processing, non-commuting sharp qubit observables by using non-orthogonal programming states.

In this work we generalize the multimeter presented in \cite{ZiBu05} and show that  a similar defeat of the limitations is possible in all finite dimensions. This generalization is achieved by studying the programmability of covariant observables. We then continue by deriving an inequality between the programming resources -- both classical and quantum -- and the corresponding programmed observables. 
This result reveals the true underlying limitations of post-processing assisted quantum programming of observables.

\section{Preliminaries}

In this preliminary section we fix the notation and recall some basic concepts used throughout this work.
We denote a complex separable Hilbert space by $\hi$, of either finite or countably infinite dimension. 
The set of bounded linear operators on $\hi$ is denoted by $\lh$, the set of trace class operators by $\trh$ and the set of projections by $\ph$. The identity element in $\lh$ is denoted by $\id_\hi$. The support of an operator $\Bo\in \lh$ is denoted by $\supp{\Bo}$.

\subsection{Quantum states and channels}

A \emph{quantum state} is represented by a positive operator $\varrho \in \trh$ with $\tr{\varrho}=1$ and the set of quantum states of $\hi$ is denoted by $\sh$. States $\varrho_1$ and $\varrho_2$ are said to be orthogonal when $\supp{\varrho_1} \cap \supp{\varrho_2} = \{0\}$. The extremal elements of $\sh$ are called \emph{pure} or equivalently \emph{vector states}, since any such element can be written as $\varrho = P_\psi$ for some unit vector $\psi \in \hi$, where $P_\psi \varphi = \ip{\psi}{\varphi}\psi$ for all $\varphi\in\hi$.

The transformations of quantum states are represented as normal linear mappings $\mc E: \trh \rightarrow \mc T(\hi')$ that are completely positive and trace-preserving: such transformations are called {\it quantum channels}. Equivalently, $\mc E$ is a quantum channel whenever its dual mapping $\mc E^*: \mc L(\hi') \rightarrow \lh$, defined via the duality relation $\tr{B \, \mc E(T)} = \tr{\mc E^*(B)\, T}$ for all $B \in \mc L(\hi'), T\in \trh$, is completely positive and unital, that is $\mc E^*(\id_{\hi'})=\id_\hi$.

\subsection{Observables and post-processing}

Let $\Omega_n = \{x_1, x_2,\ldots, x_n\}$ be a non-empty finite set.
We say that a map $\Eo: \Omega_n \rightarrow \lh$ is an $n$-valued {\it quantum observable} whenever $\Eo(x)\geq 0$ for all $x\subset\Omega_n$ and $\sum_i \Eo(x_i) = \id_\hi$, that is whenever $\Eo$ is a positive operator-valued measure (POVM). The non-zero operators $\Eo(x)\geq 0$, $x\subset\Omega$ in the range of an observable are called effects. Furthermore, we call an observable $\Ao: \Omega_n \rightarrow \lh$ {\it sharp}, or a projection-valued measure (PVM), if all the effects of $\Ao$ are projections: $\Ao(x)^2 = \Ao(x)$ for every $x \subset \Omega_n$.
We often shorten the notation by writing $\Eo(i) = \Eo(x_i)$.  

The observables on $(\hi,\Omega_n)$ form a convex set which is denoted by $\mc O(\hi,\Omega_n)$. We say that the conditional probability $\lambda: (\Omega_n, \Omega_m) \rightarrow [0,1]$ is a {\it post-processing} of $\E \in \mc O(\hi,\Omega_n)$ into $\Fo \in \mc O(\hi,\Omega_m)$ if $\lambda(x_i | \cdot)$ is a probability measure for every $x_i\in\Omega_n$ and $\Fo(y) = \sum_i \lambda(x_i | y) \E(x_i)$ for all $y\subset \Omega_m$;  in such a case we briefly write $\Fo = \lambda \star \E$. 

We call the extremal elements of a set of observables are called {\it extremal observables}. It is known, that the only to produce extremal observables from another observable by means of post-processing is by merging together effects of the other observable. In other words, if $\Eo\in\mc O(\hi, \Omega_n)$ is an observable that can be post-processed into an extremal observable $\Fo\in \mc O(\hi, \Omega_m)$ with $\lambda: (\Omega_n,\Omega_m) \rightarrow [0,1]$, then $\lambda(x_i|y) \in \{0,1\}$ for all $x_i\in \Omega_n$ and $y\subset \Omega_m$ \cite{JePu07}.

\subsection{Measurements and programmable multimeters}

A {\it measurement set-up} is mathematically given by a 4-tuple $\langle \K, \Zo,\mc V, \xi \rangle$, where $\K$ is the Hilbert space associated with the measurement apparatus, $\Zo: \Omega_n \rightarrow \lk$ is the pointer observable, $\mc V: \trhk \rightarrow \trhk$ is a quantum channel describing the measurement interaction and $\xi \in \sk$ is the initial probe state \cite{QTM96}.
The observable on $\hi$ measured in this process is given by
\begin{eqnarray}\label{eq:repro}
\Eo(x) = \ptr{\mc V^*\left(\id_\hi \otimes \Zo(x)\right) \, \id_\hi \otimes \xi} \,, \qquad x \subset \Omega_n\,.
\end{eqnarray}
In the case such as given above, we say that the observable $\Eo$ is realized by the measurement $\langle \ki, \Zo, \mc V, \xi \rangle$.

From Eq.\,\eqref{eq:repro} it is readily concluded that measurements may be altered to realise different observables by changing the initial probe state.
We will call such an action {\it quantum programming} and the programmable device, mathematically described by a triplet $\langle \K, \Zo,\mc V \rangle$, a {\it quantum multimeter} \cite{DaPe05, HeTu15, Dusek2002, FiDuFi02}. Any initial state of the probe then acts as a programming instruction for the multimeter and is thus called a {\it programming state}.

The programming states of any two unequal sharp observables are necessarily orthogonal regardless of the multimeter; see \cite{HeTu15} and Ex.\,\ref{ex:sharportho}. More generally, this orthogonality persists whenever the two programmed observables can be post-processed into different sharp observables via some fixed post-processing. Interestingly, this orthogonality need no longer hold if the post-processings are allowed to be different \cite{ZiBu05,HeTu15}, as we shall demonstrate shortly. Therefore, one can introduce post-processing as an additional programming resource, that is the programmed observable
\begin{eqnarray}
\Eo_{ij} (x) = \sum_k \lambda_j(x_k | x) \, \ptr{\mc V^*\left(\id_\hi \otimes  \Zo(x_k)\right) \, \id_\hi \otimes \xi_i} \,, \quad x\subset \Omega_m 
\end{eqnarray}
depends, not only on the state $\xi_i \in \sk$, but also on the post-processing $\lambda_j: (\Omega_n, \Omega_m) \rightarrow [0,1]$. To emphasize the difference to mere state dependent programming we will call this scenario {\it post-processing assisted quantum programming}.
 
\section{Programming of covariant observables} \label{sec:covariant}

In this section we use the framework of covariant observables to generalize the multimeter presented in \cite{ZiBu05}.
The main conclusion is that with post-processing assisted programming it is possible to realize different sharp observables with non-orthogonal programming states. 

\subsection{Structure and programming of covariant observables}

We begin by recalling some basic facts about covariant observables \cite{SSQT01,GTAQT06}.
Let $\hi$ be a finite $d$-dimensional Hilbert space.
Let $G$ be a finite group with $\# G$ elements and $U:G\to\uh$ its irreducible projective representation.
For each $\xi\in\sh$, we define a mapping $\Eo_\xi:G\to\lh$ by formula
\begin{equation}\label{eq:seed}
\Eo_\xi(g) = \frac{d}{\# G} U(g) \, \xi \, U(g)^\ast \, .
\end{equation}
The operators $\Eo_\xi(g)$ are positive and it is a direct consequence of Schur's lemma that $\sum_g \Eo_\xi(g) = c \id$ for some constant $c$. 
By calculating the traces of both sides we see that $c=1$ and hence $\Eo_\xi$ is an observable. The structure of an observable $\Eo_\xi$ is completely determined by the state $\xi$, which is called the seed operator for $\Eo_\xi$.
In the case of a pure state $\xi=P_\psi$, we also denote $\Eo_\xi=\Eo_\psi$.

We recall from \cite{DaPe05} that all observables $\Eo_\xi$ of the form of Eq.\,\eqref{eq:seed} can be programmed with a single multimeter. 
Namely, fix an orthonormal basis $\{\phi_\ell \}_{\ell=1}^d$ of $\hi$.
For each $g\in G$, we define a unit vector $u(g)\in\hi\otimes\hi$ by
\begin{equation*}
u(g) = \frac{1}{\sqrt{d}} \sum_{k,\ell} \ip{\phi_k}{U(g)\phi_\ell} \phi_k \otimes \phi_\ell = U(g) \otimes \id (\frac{1}{\sqrt{d}} \sum_\ell \phi_\ell \otimes \phi_\ell ) \, .
\end{equation*}
For each $g\in G$, we define $\Zo(g)=\tfrac{d^2}{\# G} \kb{u(g)}{u(g)}$.
A direct calculation gives
\begin{equation}
\mathrm{tr}[\Zo(g) \, \varrho \otimes \xi^T  ]= \frac{d}{ \# G} \tr{\varrho \, U(g)\, \xi \,U(g)^*} =  \tr{\Eo_\xi(g) \, \varrho } \, ,
\end{equation}
for all $\varrho$, where the transpose $\xi^T$ is defined with respect to the orthonormal basis $\{\phi_\ell \}_{\ell=1}^d$. Since any $\xi$ defines a normalized POVM via Eq.\,\ref{eq:seed}, this also shows that $\sum_g \Zo(g) =\id$, hence $\Zo$ is an observable on $\hi\otimes\hi$.
Finally, we define the partial SWAP channel $\mc V \big(A\otimes B\otimes C \big) = B\otimes A \otimes C$ for all $A,B,C \in \lh$. Then $\< \hi\otimes \hi, \Zo, \mc V \>$ is a multimeter which for any $\eta \otimes \xi^T$, $\eta \in \sh$ being arbitrary, realizes the observable $\Eo_{\xi}(g) = \frac{d}{\# G} U(g) \,\xi\, U(g)^\ast$ (see Fig.\,\ref{fig:covmultimeter}).
Interestingly, this result implies that if two unequal covariant observables, that have the same irreducible representation but different seeds $\xi_1$ and $\xi_2$, are either sharp or can be made into unequal sharp observables with a fixed post-processing, then $\xi_1$ and $\xi_2$ are necessarily orthogonal. 
\begin{figure}[t!]
\includegraphics[width=0.9\textwidth]{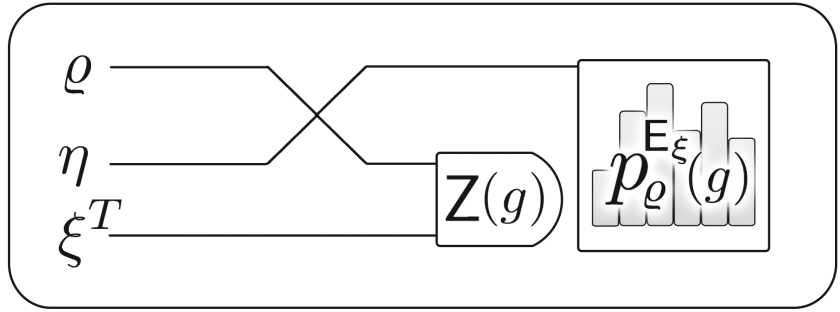}
\caption{A programmable quantum multimeter that can be used to realize all observables $\Eo_{\xi}(g) = \frac{d}{\#G} U(g) \,\xi\, U(g)^\ast$, $\xi \in \sh$, corresponding to the same irreducible projective representation $U$ with the programming states $\eta\otimes \xi^T$.}\label{fig:covmultimeter}
\end{figure}

\subsection{Obtaining sharp observables from cyclic subgroups}\label{sec:3.2}

As in the previous subsection, $U$ is an irreducible projective representation of $G$ of degree $d$.
We make an additional assumption that $G$ has a cyclic subgroup $H=\<h\>$ of order $\# G/d$.
This is obviously a restrictive assumption and allows only to use specific groups.
However, we will later demonstrate by examples that there are useful instances of this type. 

Our aim is to find sharp observables as post-processings of covariant observables. 
The method proceeds as follows:  We fix an eigenvector $\psi$ of the unitary operator $U(h)$.
For any $g\in H$, the vector $\psi$ is also an eigenvector of $U(g)$ and hence
\begin{equation}
U(g) P_\psi U(g)^\ast = P_\psi \, .
\end{equation}
It follows that
\begin{equation}
\sum_{g \in H} \Eo_{\psi}(g) =  \frac{d}{\# G} \sum_{g \in H} U(g) P_\psi U(g)^\ast = P_\psi \, .
\end{equation}
For a different left coset $g'H \neq H$ we get
\begin{equation}
\sum_{g \in g'H} \Eo_{\psi}(g) =  U(g') P_\psi U(g')^\ast\, .
\end{equation}
We conclude that by dividing $G$ into the left cosets of $H$ and forming the respective post-processing of $\Eo_{\psi}$, we obtain a sharp observable with $d$ outcomes.
 
Different eigenvectors of the unitary operator $U(h)$ may lead to different sharp observables. 
However, since the applied post-processing is determined by the subgroup $H$ and is therefore the same for all those sharp observables, these programming vectors are necessarily orthogonal. 

The described method above becomes more effective when $G$ has several cyclic subgroups $H_1=\<h_1\>,\ldots,H_n=\<h_n\>$ of the same order.
Further, we impose that any two operators $U(h_j)$ and $U(h_k)$ do not commute when $j\neq k$, so that it is guaranteed that we can choose non-orthogonal (and non-parallel) eigenvectors $\psi_1,\ldots,\psi_n$ of $U(h_1),\ldots,U(h_n)$, respectively. From a practical point of view, it is easier to start with a group that has several cyclic subgroups of the same order and then check if there exists an irreducible projective representation with the desired degree. 

\subsection{Examples}

\subsubsection*{Quaternion group}

The quaternionic group $Q_8$ consists of $8$ elements $\pm 1, \pm i , \pm j , \pm k$ satisfying the equations
\begin{equation}
i^2=j^2=k^2=ijk=-1 
\end{equation}
and $-1q=q(-1)=-q$ for each $q\in Q_8$.
There are three cyclic subgroups of order $4$: $\<i\>$, $\<j\>$ and $\<k\>$.

The previously described construction works for an irreducible projective representation of degree $8/4=2$.
We thus choose the following irreducible representation of $Q_8$:
\begin{equation*}
U(\pm 1) = \pm \id \, , \quad U(\pm i) = \pm i \sigma_x \, , \quad U(\pm j) = \mp i \sigma_y \, , \quad  U(\pm k) = \pm i \sigma_z \, .
\end{equation*}
We get three different sharp observables, the programming vectors $\psi_1,\psi_2,\psi_3$ being eigenvectors of $U(i),U(j),U(k)$, respectively.
They can be chosen such that the corresponding one-dimensional projections are:
\begin{eqnarray}
P_{\psi_1}  =  \half ( \id + \sigma_x), \quad P_{\psi_2}  =  \half (\id + \sigma_y), \quad P_{\psi_3} =  \half (\id + \sigma_z)\,.
\end{eqnarray}
The sharp observables obtained by this post-processing assisted programming are hence the complementary qubit observables consisting of the projections $\half (\id \pm \sigma_x)$, $\half (\id \pm \sigma_y)$ and $\half (\id \pm \sigma_z)$, respectively.
We have thus reproduced the result first presented in  \cite{ZiBu05}.

\subsubsection*{Finite phase space}

Let $d$ be a prime and let $\integer_d$ be the cyclic group of integers modulo $d$.
The product group $\integer_d \times \integer_d$ can be consider as a finite phase space \cite{GiHoWo04}.
It has $d+1$ cyclic subgroups of order $d$;
the generating elements of these cyclic subgroups can be chosen to be $(0,1)$ and $(1,0),(1,1),\ldots,(1,d-1)$.

The previously described construction works for an irreducible projective representation of degree $d^2/d=d$.
We fix an orthonormal basis $\{\varphi_k\}_{k\in\integer_d}$ of a $d$-dimensional Hilbert space $\hi$. 
We denote $\omega= e^{2\pi i /d}$ and define the following projective unitary representation $U$ of $\integer_d \times \integer_d$ in $\hi$:
\begin{equation}
U(x,y) \varphi_k =  \omega^{yk} \varphi_{k+x} \, . 
\end{equation}
The programming vectors $\psi_0,\psi_1,\ldots,\psi_{d-1}$ are eigenvectors of the operators $U(0,1)$ and $U(1,k), k=0,\ldots, d-1$, respectively.
They can be chosen to be:
\begin{eqnarray}
\psi_0 & = &  \varphi_0 \nonumber \\
\psi_k &=& \frac{1}{\sqrt{d}} \sum_{i=0}^{d-1} c^{(k)}_j \varphi_j, 
\end{eqnarray}
where $c^{(k)}_j=\omega^{\frac{1}{2}  jk (j-1) -j}$, $j,k = 0, \ldots , d-1$. 
In particular, we have $\mo{\ip{\psi_j}{\psi_k}}=1/\sqrt{d}$ for $j\neq k$. 
It follows from Sec.\,\ref{sec:3.2} that each of these eigenvectors will lead to a sharp observable via appropriate post-processing.

\section{Limitations on quantum programming}

Not much is known about the mutual dependencies of the programming states of general quantum observables, except those of sharp observables. Such scenarios are important since often the implemented observables differ from sharp ones. On the other hand, to our best knowledge the dependencies of post-processings and the programmed observables in post-processing assisted quantum programming have not been studied before. In this section we address these issues by introducing an inequality between measures of ``closeness'' of the programming resources and the corresponding programmed observables. We note that similar analysis was recently made for programming of quantum channels in \cite{Tukiainen2016}.

\subsection{Inequality for programming of observables}

A commonly used measure of distinguishability of quantum states is fidelity $F$. There are many equivalent ways to express fidelity, two of which we will recall next for the later use. Firstly, the theory of Uhlmann states that 
\begin{eqnarray}
F(\varrho_1, \varrho_2) = \max_{\varphi_1, \varphi_2} |\ip{\varphi_1}{\varphi_2}|,
\end{eqnarray} where the maximum is taken over all purifications $\varphi_i$ of $\varrho_i$, $i=1,2$ \cite{NielsenInfo}. Secondly, fidelity may also be recast as
\begin{eqnarray}
F(\varrho_1, \varrho_2) = \min_{\E\in \mc O(\hi)} \sum_i \tr{\E(i) \, \varrho_1}^{1/2} \tr{\E(i) \, \varrho_2}^{1/2},
\end{eqnarray} for all $\varrho_1,\varrho_2\in\sh$, where the minimum is taken over all observables $\E: \Omega_n \rightarrow \lh$ \cite{NielsenInfo}.
The latter formulation is closely related to so-called {\it Bhattacharyya coefficient} $B$ that quantifies the divergence between two (discrete) probability measures: $B(p,q) := \sum_i \sqrt{p(i) q(i)} $, where $p,q: \Omega_n \rightarrow [0,1]$, $\sum_{i=1}^n p(i) = \sum_{i=1}^n q(i) =1$.

From this point onwards we will use the shorthand notion $p^\Eo_\varrho(x) = \tr{\Eo(x) \, \varrho}$ and write $p^\Eo_\varphi$ if $\varrho=P_\varphi$. With these tools and definitions we are ready to formulate the results of this subsection.

\begin{proposition}\label{prop:mainobs}
Suppose two observables $\Eo_1$ and $\Eo_2 \in \mc O(\hi,\Omega_n)$ can be programmed with states $\xi_1$ and $\xi_2$, respectively. 
Then 
\begin{eqnarray}
F(\varrho_1, \varrho_2)F(\xi_1, \xi_2) \leq B(p^{\Eo_1}_{\varrho_1},p^{\Eo_2}_{\varrho_2})\end{eqnarray} 
for all $\varrho_1, \varrho_2 \in \sh$.
\end{proposition}
\begin{proof} 
Fix states $\varrho_1, \varrho_2 \in \sh$ and let $\varphi_1,\varphi_2 \in \hi' \otimes \hi$ be any purifications of them, respectively. Furthermore, let $\phi_1, \phi_2 \in \ki\otimes\ki'$ be the any purifications of the programming states $\xi_1,\xi_2$, respectively.
Then in the total Hilbert space $\hi' \otimes \hi \otimes \ki \otimes \ki'$ we have
\begin{eqnarray}
&&|\langle \varphi_1\otimes \phi_1 | \varphi_2\otimes \phi_2 \rangle | 
=| \sum_i \langle\varphi_1\otimes \phi_1 | \id \otimes \mc V^*( \id \otimes \Z(i) ) \otimes \id |\varphi_2\otimes \phi_2 \rangle | 
\nonumber \\
&&\leq \sum_i | \langle  \varphi_1\otimes \phi_1 | \id \otimes \mc V^*( \id \otimes \Z(i) ) \otimes \id |\varphi_2\otimes \phi_2 \rangle | \nonumber \\
&&\leq \sum_i \langle  \varphi_1\otimes \phi_1 |\id \otimes \mc V^*( \id \otimes \Z(i) ) \otimes \id | \varphi_1\otimes \phi_1 \rangle^{1/2} \times \nonumber \\
&&\times \langle \varphi_2\otimes \phi_2 |\id \otimes \mc V^*( \id \otimes \Z(i) ) \otimes \id |
\varphi_2\otimes \phi_2 \rangle^{1/2} \nonumber \\
&& = \sum_i \tr{\mc V^*( \id \otimes \Z(i) ) \, \varrho_1\otimes \xi_1}^{1/2}   \, \tr{ \mc V^*( \id \otimes \Z(i) )\, \varrho_2\otimes \xi_2}^{1/2} \nonumber \\
&&=
\sum_i \tr{\E_1 (i)\,\varrho_1}^{1/2} \tr{\E_2 (i) \,\varrho_2}^{1/2}.
\end{eqnarray}
Since
\begin{eqnarray}
\max_{\varphi_1, \varphi_2} \max_{\phi_1, \phi_2} |\langle \varphi_1\otimes \phi_1 | \varphi_2\otimes \phi_2 \rangle | = F(\varrho_1, \varrho_2) \, F(\xi_1, \xi_2),
\end{eqnarray}
we conclude the proof by taking the maxima over all the possible purifications in the above inequality. 
\qed \end{proof} 

Motivated by the previous proposition, we may define a quantity $B(\E_1,\E_2) = \inf_{\varrho_1,\varrho_2} \frac{B(p^{\E_1}_{\varrho_1},p^{\E_2}_{\varrho_2})}{F(\varrho_1,\varrho_2)}$ which sets the relation $F(\xi_1,\xi_2) \leq B(\E_1,\E_2)$ between the programmed observables $\Eo_i$ and the programming states $\xi_i$, $i=1,2$. We wish to point out that due to the double infima taken over the state space $B(\E_1,\E_2)$ is always well defined, although $\frac{B(p^{\E_1}_{\varrho_1},p^{\E_2}_{\varrho_2})}{F(\varrho_1,\varrho_2)}$ can be indefinite in general as the denominator may vanish.  Below we have listed some of its properties.
\begin{proposition}\label{prop:obsfid}
Let $\E_1, \E_2 \in \mc O(\hi, \Omega_n)$. The quantity $B(\E_1,\E_2)$ satisfies the following properties:
\begin{itemize}
\item[(B1)] $B(\E_1,\E_2)=B(\E_2,\E_1)$,
\item[(B2)] $B(\E_1,\E_2)\in [0,1]$,
\item[(B3)] $B(\E_1,\E_2)=1$ iff $\E_1=\E_2$,
\item[(B4)] $B(\E_1,\E_2)=B(U^* \E_1 U, U^* \E_2 U)$, for all unitaries $U$,
\item[(B5)] $B(\E_1,\E_2) \leq B(\mc E^*(\E_1),\mc E^*(\E_2))$,
\item[(B6)] $B(\E_1,\E_2) \leq B(\lambda\star\E_1,\lambda\star\E_2)$,
\end{itemize}
for all channels $\mc E^*:\mc L(\hi)\rightarrow\mc L(\hi')$ and post-processings $\lambda: (\Omega_n, \Omega_m) \rightarrow [0,1]$.
\end{proposition}
\begin{proof}
The first and the fourth properties follow trivially from the symmetry and unitary invariance of the fidelity of quantum states. The second one can be proved by noticing that $\inf_{\varrho_1,\varrho_2} \frac{B(p^{\E_1}_{\varrho_1},p^{\E_2}_{\varrho_2})}{F(\varrho_1,\varrho_2)}\leq \frac{B(p^{\E_1}_{\varrho},p^{\E_2}_{\varrho})}{F(\varrho,\varrho)} \leq 1 $, for all $\varrho\in\sh$, where the latter limit follows immediately from the Cauchy-Schwarz inequality.

To prove (B3), we first suppose that $B(\E_1,\E_2) = 1$, whilst supposing that $\E_1\neq \E_2$. Then there exists a state $\varrho$ such that $p^{\E_1}_\varrho\neq p^{\E_2}_\varrho$, implying that $1=B (p^{\E_1}_\varrho, p^{\E_2}_\varrho) < 1$. Therefore, $B(\E_1,\E_2) \leq \frac{B (p^{\E_1}_\varrho, p^{\E_2}_\varrho)}{F(\varrho, \varrho)} < 1$. This contradiction proves the implication $B(\E_1,\E_2) = 1 \Rightarrow \E_1 = \E_2$. On the other hand, suppose that $\E_1 = \E_2 = \E$. Then 
\begin{eqnarray}
 B(p^{\E_1}_{\varrho_1},p^{\E_2}_{\varrho_2}) &=& \sum_i \tr{\E(i) \,\varrho_1}^{1/2} \tr{\E(i) \, \varrho_2} ^{1/2}  \nonumber \\
 &\geq & \min_{\E}  \sum_i \tr{\E(i) \,\varrho_1}^{1/2} \tr{\E(i) \, \varrho_2} ^{1/2}  \nonumber \\
 & = & F(\varrho_1, \varrho_2).
\end{eqnarray}
Hence, $\frac{B(p^{\E_1}_{\varrho_1},p^{\E_2}_{\varrho_2})}{F(\varrho_1,\varrho_2)} \geq 1$ which, together with (B2), proves that $B(\E_1,\E_2)=1$.

We next prove (B5). Suppose there exists a Hilbert space $\ki'$ where $\mc E^*$ admits a unitary dilation, $\mc E^*(B) = \text{tr}_{\ki'}[ U^* ( B \otimes \id_{\ki'}) U \, \id_\hi \otimes \eta]$ for all $B\in\lh$ and some fixed unitary $U$ on $\hi \otimes \ki'$ and fixed state $\eta \in \mc S(\ki')$. Then
\begin{eqnarray}
& &B(p^{\mc E^* (\E_1)}_{\varrho_1},p^{\mc E^* (\E_2)}_{\varrho_2}) \nonumber \\
&=& \sum_i \tr{U^* ( \E_1(i) \otimes \id_{\ki'}) U \,  \varrho_1 \otimes \eta}^{1/2} \tr{\text{tr}_{\ki'}[U^* ( \E_2(i) \otimes \id_{\ki'}) U \, \varrho_2 \otimes \eta}^{1/2} \nonumber \\
& = & B(p^{U^*(\E_1 \otimes \id_{\ki'} ) U}_{\varrho_1\otimes \eta},p^{U^*(\E_2 \otimes \id_{\ki'} ) U}_{\varrho_2\otimes \eta}).
\end{eqnarray}
Therefore
\begin{eqnarray}
B(\mc E^*(\E_1),\mc E^*(\E_2)) &= & \inf_{\varrho_1,\varrho_2} \frac{B(p^{U^*(\E_1 \otimes \id_{\ki'} ) U}_{\varrho_1\otimes \eta},p^{U^*(\E_2 \otimes \id_{\ki'} ) U}_{\varrho_2\otimes \eta})}{F(\varrho_1\otimes \eta, \varrho_2\otimes \eta)} \nonumber \\
& \geq & B(U^*(\E_1 \otimes \id_{\ki'} ) U, U^*(\E_2 \otimes \id_{\ki'} ) U) \nonumber \\
&=& B(\E_1 \otimes \id_{\ki'} , \E_2 \otimes \id_{\ki'}) \nonumber \\
&= & \inf_{\chi_1,\chi_2} \frac{\sum_i \tr{\E_1(i) \otimes \id_{\ki'} \, \chi_1}^{1/2} \tr{\E_2(i) \otimes \id_{\ki'} \, \chi_2}^{1/2} }{F(\chi_1,\chi_2)} \nonumber \\
&\geq & \inf_{\chi_1,\chi_2} \frac{\sum_i \tr{\E_1(i) \, \text{tr}_{\ki'}[\chi_1]}^{1/2} \tr{\E_2(i) \, \text{tr}_{\ki'}[\chi_2]}^{1/2} }{F(\text{tr}_{\ki'}[\chi_1],\text{tr}_{\ki'}[\chi_2])}
\nonumber \\
&=& B(\E_1, \E_2),
\end{eqnarray}
where we have used the properties (B4) and the monotonicity of the fidelity, $F(\varrho_1, \varrho_2) \leq F(\mc C(\varrho_1),\mc C(\varrho_2))$ for all channels $\mc C:\mc T(\hi) \rightarrow \mc T( \hi'')$.

Finally, we prove (B6). Fix any post-processing $\lambda:(\Omega_n, \Omega_m) \rightarrow [0,1]$. Since $\sum_i \lambda(i | j) =1$ for all $j$ we have
\begin{eqnarray}
B(p^{\E_1}_{\varrho_1},p^{\E_2}_{\varrho_2}) &=& \sum_j \sum_i \lambda(i | j) \, \tr{\E_1(j) \,\varrho_1}^{1/2} \tr{\E_2(j) \, \varrho_2} ^{1/2} \nonumber \\
&\leq &  \sum_i \left(\sum_j \lambda(i | j) \, \tr{\E_1(j) \,\varrho_1}\right)^{1/2} \left(\sum_j \lambda(i | j) \, \tr{\E_2(j) \,\varrho_2}\right) ^{1/2} \nonumber \\
&= & B(\lambda\star p^{\E_1}_{\varrho_1},\lambda\star p^{\E_2}_{\varrho_2}),
\end{eqnarray} 
where the inequality follows from the Cauchy-Schwarz inequality. 
\qed \end{proof}

For $B(\E_1,\E_2)$ to be a satisfactory upper bound for the fidelity of the programming states, we should be able to bind it to some known results on the programmability of quantum observables. The following example establishes such a connection.

\begin{example}\label{ex:sharportho}
Suppose that two different sharp observables $\A_1$ and $\A_2 \in \mc O(\hi, \Omega_n)$ can be programmed with states $\xi_1$ and $\xi_2$ respectively. Then $F(\xi_1, \xi_2)=0$ or equivalently $\xi_1$ and $\xi_2$ are orthogonal. In verifying this claim we follow the proof of Prop.\,4 in \cite{HeTu15}. Let us fix disjoint sets $x,y \subset \Omega_n$ such that condition $\A_1(x) \A_2(y) \neq 0$ holds. Then there exist unit vectors $\varphi_1$ and $\varphi_2$ such that $\A_1(x)\varphi_1 = \varphi_1,$ $\A_2 (y) \varphi_2 = \varphi_2$ and $\ip{\varphi_1}{\varphi_2}\neq 0$. Furthermore, $\A_1(x^c) \varphi_1 = 0$ and $\A_2(y^c) \varphi_2 =0$. Hence, $p^{\A_1}_{\varphi_1}(y)=\ip{\varphi_1}{\A_1(y) \, \varphi_1}=0,$ $p^{\A_2}_{\varphi_2}(x)=0$ and $p^{\A_i}_{\varphi_i}((x\cup y)^c)=0,$ for $i=1,2$, so that we have
\begin{eqnarray}
B (p^{\A_1}_{\varphi_1}, p^{\A_2}_{\varphi_2})  &=& \sum_{x_i\in x}  \tr{\A_1(x_i) \, P_{\varphi_1}}^{\frac{1}{2}} \tr{\A_2(x_i) \, P_{\varphi_2}}^{\frac{1}{2}} \nonumber \\
&& +  \sum_{x_i\in y}  \tr{\A_1(x_i) \, P_{\varphi_1}}^{\frac{1}{2}} \tr{\A_2(x_i) \, P_{\varphi_2}}^{\frac{1}{2}} \nonumber \\
&&+ \sum_{x_i\in (x\cup y)^c}  \tr{\A_1(x_i) \, P_{\varphi_1}}^{\frac{1}{2}} \tr{\A_2(x_i) \, P_{\varphi_2}}^{\frac{1}{2}} \\
&=&0.
\end{eqnarray}
Since $F(\varphi_1, \varphi_2) = |\ip{\varphi_1}{\varphi_2}|\neq 0$, we have $F(\xi_1,\xi_2) \leq \frac{B (p^{\A_1}_{\varphi_1}, p^{\A_2}_{\varphi_2})}{F(\varphi_1, \varphi_2)} = 0$ which proves the claim. The above also shows that $B(\A_1,\A_2)=0$ for all unequal sharp observables $\Ao_1$ and $\Ao_2$. This example implies that no state programmable multimeter $\< \ki, \Zo, \mc V\>$ described by a separable Hilbert space $\ki$ can realize all observables, since there is only a countable number of orthogonal states in $\sk$.
\end{example}

\subsection{Inequality for post-processing assisted programming of observables}

As noted in Sect.\,\ref{sec:covariant}, post-processing can be an advantageous resource in quantum programming. In \cite{HeTu15} it was left as an open question whether it is possible to realize all observables if in addition to programming we allow for post-processing. Using the tools developed above we provide a negative answer to this question.

Following the above considerations, we define the ``closeness'' of two post-processings $\lambda_1$ and $\lambda_2:(\Omega_n, \Omega_m) \rightarrow [0,1]$ via 
\begin{eqnarray} 
F(\lambda_1, \lambda_2) &= & \inf_{p \in {\bf Prob}(\Omega_n)} B(\lambda_1 \star p, \lambda_2 \star p) \nonumber \\
&=& \inf_{p \in {\bf Prob}(\Omega_n)} \sum_j \sqrt{\sum_i \lambda_1(i|j) \, p(i)} \sqrt{\sum_i \lambda_2(i|j) \, p(i)} \, , \nonumber
\\
\end{eqnarray}
however, $F(\lambda_1, \lambda_2)$ can be more easily calculated using the following lemma. 
\begin{lemma}\label{lem:min}
$F(\lambda_1, \lambda_2) = \min_{i\in\{1,\ldots,n\} } \sum_j \sqrt{\lambda_1(i|j) } \sqrt{\lambda_2(i|j)}$.
\end{lemma}

\begin{proof}
Trivially $F(\lambda_1, \lambda_2) \leq \min_{i\in\{1,\ldots,n\} } \sum_j \sqrt{\lambda_1(i|j) } \sqrt{\lambda_2(i|j)},$ since $\delta_i(x_i)=1,$ $\delta_i(\Omega_n\setminus x_i)=0$ are probability measures on $\Omega_n$ for every $i$. On the other hand, using the Cauchy-Schwarz inequality, we confirm that 
\begin{eqnarray} 
F(\lambda_1, \lambda_2) &\geq & \inf_{p \in {\bf Prob}(\Omega_n)} \sum_i p(i) \sum_j \sqrt{ \lambda_1(i|j) \, \lambda_2(i|j) }  \nonumber \\
&\geq & \min_{i\in\{1,\ldots,n\} } \sum_j \sqrt{ \lambda_1(i|j) \, \lambda_2(i|j) }. 
\end{eqnarray}
\qed \end{proof}

The next proposition is the main result of this section.
\begin{proposition}\label{prop:postp1} 
Suppose two observables $\Eo_1$ and $\Eo_2 \in \mc O(\hi, \Omega_n)$ can be programmed with states $\xi_1$ and $\xi_2$, respectively. 
Then for any post-processings $\lambda_1$ and $\lambda_2:(\Omega_n, \Omega_m) \rightarrow [0,1]$ 
\begin{eqnarray}
F(\varrho_1, \varrho_2)F(\xi_1, \xi_2)F(\lambda_1, \lambda_2) \leq B(\lambda_1 \star p^{\Eo_1}_{\varrho_1},\lambda_2 \star p^{\Eo_2}_{\varrho_2}), 
\end{eqnarray} for all $\varrho_1, \varrho_2 \in \sh$.
\end{proposition}

\begin{proof}
By the previous Lemma\,\ref{lem:min} we have $F(\lambda_1, \lambda_2) \leq \sum_j \sqrt{\lambda_1(i|j) } \sqrt{\lambda_2(i|j)}$ for all $i=1,\ldots,n$. From Prop.\,\ref{prop:mainobs} we already know $F(\varrho_1, \varrho_2)F(\xi_1, \xi_2) \leq B(p^{\E_1}_{\varrho_1}, p^{\E_2}_{\varrho_2})$. Therefore,
\begin{eqnarray}
& &F(\varrho_1, \varrho_2)F(\xi_1, \xi_2)F(\lambda_1, \lambda_2) \nonumber \\
&\leq & \sum_i \sum_j   \left(\lambda_1(i|j) \, \tr{\E_1(i) \,\varrho_1}\right)^{1/2} \left(\lambda_2(i|j) \, \tr{\E_2(i) \,\varrho_2}\right)^{1/2} \nonumber \\
&\leq & \sum_i \left(  \sum_j \lambda_1(i|j) \, \tr{\E_1(i) \,\varrho_1} \right)^{1/2} \, \left(  \sum_j \lambda_2(i|j) \, \tr{\E_2(i) \,\varrho_2} \right)^{1/2} \nonumber \\
&=& B(\lambda_1 \star p^{\E_1}_{\varrho_1},\lambda_2 \star p^{\E_2}_{\varrho_2}),
\end{eqnarray}
where the second estimate is due to the Cauchy-Schwarz inequality.
\qed \end{proof}

It follows from the previous proposition that $F(\xi_1, \xi_2)F(\lambda_1, \lambda_2) \leq B(\lambda_1 \star \E_1,\lambda_2 \star \E_2)$. Suppose now that it is possible to realize two different sharp observables via post-processing assisted programming. From Ex.\,\ref{ex:sharportho} it can be concluded then that either $F(\xi_1, \xi_2)=0$ or $F(\lambda_1, \lambda_2)=0$. Assume that the Hilbert space $\ki$ of the multimeter $\< \ki, \Zo, \mc V\>$ is separable and the pointer observable $\Zo$ is $m$-valued. Then one can use post-processing assisted programming to realize at most $\dim (\ki) \cdot \text{Bin}\big(m,n\big)$ sharp observables with $n$-outcomes with $\< \ki, \Zo, \mc V\>$. Here $\text{Bin}\big(m,n\big)$ is the binomial coefficient $\text{Bin}\big(m,n\big)= \frac{m!}{n!(m-n)!}$ which characterizes the size of the set $\{\lambda_i: (\Omega_n, \Omega_m) \rightarrow [0,1] \,|\, F(\lambda_i, \lambda_j) =0\}$. This proves that, with the above assumptions, even post-processing assisted programming cannot universally realize all quantum observables comprising an uncountable set.

We end the section by pointing out that the trade-off between the fidelities $F(\xi_1, \xi_2)$ and $F(\lambda_1, \lambda_2)$ when programming sharp observables is more subtle than what is captured by $F(\xi_1, \xi_2) F(\lambda_1, \lambda_2) = 0$. Let us recall, that in the quaternion example of Sect.\,\ref{sec:covariant} we considered the post-processing assisted programming of three sharp spin-observables where, in particular, the programming states $P_{\psi_i}$ satisfy $F(P_{\psi_i},P_{\psi_j})=1/\sqrt{2}$, $i\neq j$. It turns out that the choice of states $P_{\psi_i}$, $i=1,2,3$, is optimal in the following sense.
\begin{figure}[t!]
\includegraphics[width=1.0\textwidth]{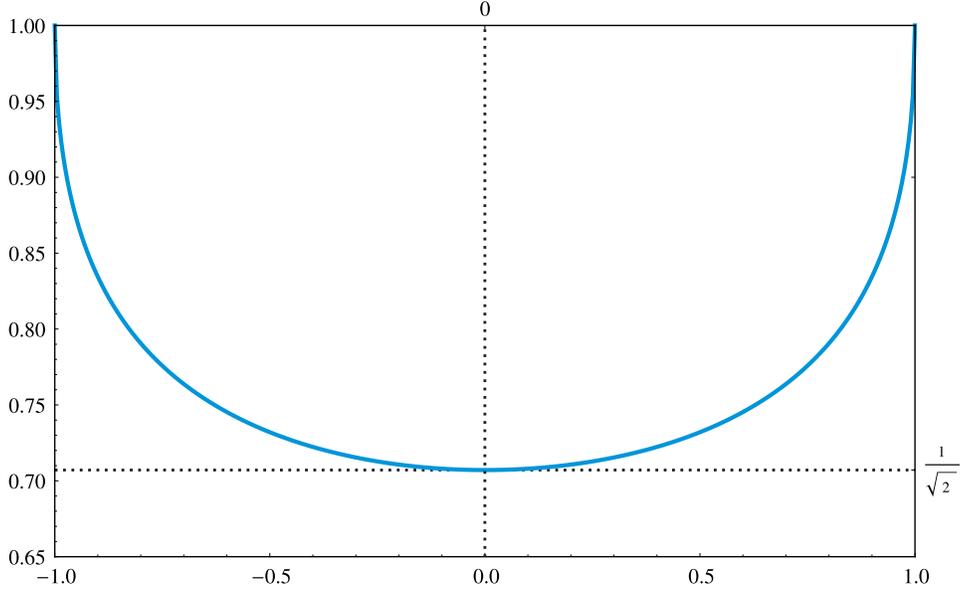}
\caption{A universally valid upperbound for the fidelity of programming states $\xi_1$ and $\xi_2$ in post-processing assisted programming of sharp spin-observables $\Ao_i(\pm)=\frac{1}{2}\left(\id_{\C^2} \pm \vec{a}_i \cdot \vec{\sigma} \right)$ versus $\vec a_1 \cdot \vec a_2$. The bound, due to the right-hand-side of Ineq.\,\eqref{eq:sharpmin}, attains its minimum $1/\sqrt{2}$ for $\vec a_1 \cdot \vec a_2 = 0$. }\label{fig:combi}
\end{figure}
\begin{example}
Let $\xi_1$ and $\xi_2$ be the programming states of any observables $\E_1$ and $\E_2:\Omega_m \rightarrow \lh$ that can be post-processed into sharp spin-observables $\Ao_i(\pm)=\frac{1}{2}\left(\id_{\C^2} \pm \vec{a}_i \cdot \vec{\sigma} \right)$, $\vec a_i \in \R^3,$ $||\vec a_i ||=1$, $i=1,2$, respectively. Due to the extremality of $\Ao_i$ there exist $x,y\subset\Omega_m$ such that $\Ao_1(+) = \Eo_1(x)$ and $\Ao_2(+) = \Eo_2(y)$ \cite{JePu07}. In particular, for the set $\Omega = \{x\cap y, x\cap y^c, x^c \cap y, x^c \cap y^c \}$ the following table holds

\begin{table}[h!]
\centering
\label{tab:sets}
\begin{tabular}{l|cc}
               & $\Eo_1(\cdot)$                		& $\Eo_2(\cdot)$      \\ \hline
$\phantom{^c}x\cap y$      & $\phantom{^c}p_+(y) \Ao_1(+)$		& $\phantom{^c}q_+(x) \Ao_2(+)$     \\
$\phantom{^c}x\cap y^c$    & $p_+(y^c) \Ao_1(+)$				& $\phantom{^c}q_-(x) \Ao_2(-)$     \\
$x^c \cap y$   & $\phantom{^c}p_-( y) \Ao_1(-)$     & $q_+(x^c) \Ao_2(+)$ \\
$x^c \cap y^c$ & $p_-(y^c) \Ao_1(-)$				& $q_-(x^c) \Ao_2(-)$
\end{tabular}
\end{table}
\noindent where $p_\pm$ and $q_\pm:\Omega_m \rightarrow [0,1]$ are some probability measures. Defining a merging type post-processing $\lambda:(\Omega_m, \Omega) \rightarrow [0,1]$ according to the previous table we have 
$$
F(\xi_1,\xi_2) \leq B(\lambda\star \Eo_1,\lambda\star \Eo_2) \leq B(\lambda\star p^{\Eo_1}_{\varrho_1},\lambda\star p^{\Eo_2}_{\varrho_2})/F(\varrho_1,\varrho_2)
$$
for all states $\varrho_1,\varrho_2$ satisfying $F(\varrho_1,\varrho_2)\neq 0$. By choosing these states appropriately, we get the following inequality
\begin{eqnarray}\label{eq:sharpmin}
&& F(\xi_1,\xi_2) \nonumber \\ &\leq& \sqrt{2}  \min \left(  \frac{ \sqrt{p_+(y) q_+(x)}}{\sqrt{1+ \vec{a}_1 \cdot \vec{a}_2}} ,\frac{\sqrt{ p_+(y^c) q_-(x)}}{\sqrt{1-\vec{a}_1 \cdot \vec{a}_2 }}, \frac{\sqrt{ p_-( y) q_+(x^c)}}{\sqrt{1-\vec{a}_1 \cdot \vec{a}_2 }},\frac{\sqrt{ p_-(y^c) q_-(x^c)}}{\sqrt{1+\vec{a}_1 \cdot \vec{a}_2 }} \right) \nonumber \\
\end{eqnarray}
For example, the first term follows from choosing $\varrho_i = \Ao_i(+)$, $i=1,2$.
By maximizing the right-hand-side over $p_\pm$ and $q_\pm$ we get a universally valid upper bound for programming states $\xi_1$ and $\xi_2$ in post-processing assisted quantum programming, plotted in Fig.\,\ref{fig:combi} as a function of $\vec a_1 \cdot \vec a_2$.
For the orthogonal spin directions, $\vec a_1 \cdot \vec a_2 = 0$, the upper bound reaches its minimum $1/\sqrt{2}$.
\end{example}

\section{Summary and outlook}

The upside of designing a programmable device instead of a single purpose machine is, that it is more economical on resources. Accordingly, programmable multimeters have applications in schemes where one wants to change measurements on the fly, such as in quantum-state discrimination, in quantum cryptography and eavesdropping, and in measurement based quantum computation, to name a few. Although the exact universality of a programmable multimeter is known to be unachievable, it is obviously desirable to have one that is as optimal as possible (for a given task). 

Related to the above, we showed that it is advantageous to include classical post-processing as a part of quantum programming scenarios. In addition, we revealed a relation between the programmed observables and the corresponding programming and post-processing resources. This relation allowed us to prove that even post-processing assisted programming cannot universally realize all quantum observables. We believe that our survey establishes a mathematical basis for further investigation of advantages and limitations of the three processes: classical post-processing, quantum programming and the combined protocol of post-processing assisted programming.

We will finish with a short discussion about one concrete application of post-processing assisted programming. In Ref.\,\cite{Dusek2002} the question ``What is the form of the optimal multimeter to approximate within a fixed error all observables in a given Hilbert space?'' was asked, however, in lack of a proper distance measure of observables the authors left this question open. Consequently, a programming scenario providing a partial answer to this question was proposed in Ref.\,\cite{DaPe05}, and it was soon shown to be the optimal one in  Ref.\,\cite{Perez06}. This scenario, however, only concentrates on approximating observables with fixed number of outcomes and therefore does not answer the original problem in its full generality. In vague terms, one can attempt to solve the problem by programming an optimal set of observables, the elements of which could then be classically post-processed to approximate any observable. We believe that the distance measure of observables $B(\Eo_1,\Eo_2)$ we have introduced can serve as a good candidate for answering such optimality questions, in particular due to its direct relation with the corresponding programming and post-processing resources. However, to answer the above optimality question in the fully general setting will likely require an intricate interplay between programming and post-processing -- the details of which are beyond the scope of this work -- and the question of Ref.\,\cite{Dusek2002} is hence left open for future investigation.

\section*{Acknowledgements}

M.T. acknowledges financial support from the University of Turku Graduate School (UTUGS).
The authors are grateful to Juha-Pekka Pellonp\"a\"a and Tom Bullock for their comments on the manuscript.



\end{document}